\documentclass[12pt]{article}
\usepackage[left=35mm,right=35mm,top=35mm,bottom=35mm]{geometry}
\usepackage{amsmath}
\usepackage{amsthm}
\usepackage{amssymb}
\usepackage{amsfonts}
\usepackage{mathrsfs}
 
\DeclareMathOperator*{\slim}{s-lim}
\DeclareMathOperator*{\supp}{supp}

\newcommand{\Eqn}[1]{&\hspace{-0.5em}#1\hspace{-0.5em}&}
\makeatletter

\@addtoreset{equation}{section}
\makeatother
\newtheorem{The}{Theorem}[section]
\newtheorem{Ass}[The]{Assumption}

\newtheorem{Prop}[The]{Proposition}
\newtheorem{Cor}[The]{Corollary}
\newtheorem{Lem}[The]{Lemma}

\title{Inverse scattering in Stark effect}
\author{Atsuhide ISHIDA\\
\\
Department of Liberal Arts, Faculty of Engineering,
 \\Tokyo University of Science\\
\normalsize 3-1 Niijuku, 6-chome, Katsushika-ku,Tokyo 125-8585, Japan\\
\normalsize E-mail: aishida@rs.tus.ac.jp
}
\date{}

\begin{document}
\maketitle

\begin{abstract}
\noindent
We study one of the multidimensional inverse scattering problems for quantum systems governed by the Stark Hamiltonians. By applying the time-dependent method developed by Enss and Weder in 1995, we prove that the high-velocity limit of the scattering operator determines uniquely the short-range interaction potentials. Moreover, we prove that, when a long-range interaction potential is given, the high-velocity limit of the Dollard-type modified scattering operator determines uniquely the short-range part of the interactions. We allow the potential functions to belong to very broad classes. These results are improvements on the previous results obtained by Adachi and Maehara in 2007 and Adachi, Fujiwara, and Ishida in 2013.
\end{abstract}


\section{Introduction\label{introduction}}
We investigate one of the inverse scattering problems involving a quantum system in an electric field $E=e_1=(1,0,\ldots,0)\in\mathbb{R}^n$. Throughout this paper, we assume that the space dimension $n\geqslant2$. The free dynamics is described by the following Stark Hamiltonian, a self-adjoint operator acting on $L^2(\mathbb{R}^n)$,
\begin{equation}
H_0^{\rm S}=|p|^2/2-x_1,\label{free}
\end{equation}
where $p=-{\rm i}\nabla_x$ is the momentum operator with ${\rm i}=\sqrt{-1}$ thus $|p|^2=-\Delta_x=-\sum_{j=1}^n\partial_{x_j}^2$ is the negative of the Laplacian, and $x_1$ is the $1$st component of $x=(x_1,\ldots,x_n)\in\mathbb{R}^n$. The pairwise interaction potential $V$ is the multiplication operator of the real-valued function $V(x)$, the value vanishes at large distance. In more detail, $V$ is represented by a sum of parts $V^{\rm vs}\in\mathscr{V}^{\rm vs}$, $V^{\rm s}\in\mathscr{V}^{\rm s}$, and $V^{\rm l}\in\mathscr{V}_{\rm G}^{\rm l}\cup\mathscr{V}_{\rm D}^{\rm l}$ such that
\begin{equation}
V(x)=V^{\rm vs}(x)+V^{\rm s}(x)+V^{\rm l}(x),
\end{equation}
where the classes of the real-valued functions, $\mathscr{V}^{\rm vs}$, $\mathscr{V}^{\rm s}$, $\mathscr{V}_{\rm G}^{\rm l}$, and $\mathscr{V}_{\rm D}^{\rm l}$, satisfy Assumption \ref{ass} below. Prior to stating this assumption, we provide some notation. The Kitada bracket of $x$ has the usual definition, $\langle x\rangle=\sqrt{1+|x|^2}$. $F(\cdots)$ is the characteristic function of the set $\{\cdots\}$, and $\|\cdot\|$ denotes the operator norm in $L^2(\mathbb{R}^n)$.

\begin{Ass}\label{ass}
$V^{\rm vs}\in\mathscr{V}^{\rm vs}$ is decomposed into
\begin{equation}
V^{\rm vs}(x)=V_1^{\rm vs}(x)+V_2^{\rm vs}(x),
\end{equation}
where a singular part $V_1^{\rm vs}$ is $|p|^2/2$-bounded with its relative bound less than 1, $x_1V_1^{\rm vs}$ is $|p|^2/2$-bounded, a regular part $V_2^{\rm vs}$ is bounded, and $V^{\rm vs}$ satisfies
\begin{equation}
\int_0^\infty\|V^{\rm vs}(x)\langle p\rangle^{-2}F(|x|\geqslant R)\|dR<\infty.\label{enss_condition1}
\end{equation}
This decay condition \eqref{enss_condition1} is equivalent to
\begin{equation}
\int_0^\infty\|F(|x|\geqslant R)V^{\rm vs}(x)\langle p\rangle^{-2}\|dR<\infty\label{enss_condition2}
\end{equation}
{\rm(}see {\rm Reed-Simon \cite{ReSi})}.\par
$V^{\rm s}\in\mathscr{V}^{\rm s}$ belongs to $C^1(\mathbb{R}^n)$ and satisfies
\begin{equation}
|V^{\rm s}(x)|\leqslant C\langle x\rangle^{-\gamma},\quad|\partial_x^\beta V^{\rm s}(x)|\leqslant C_\beta\langle x\rangle^{-1-\alpha}\label{short_range}
\end{equation}
for the multi-index $\beta$ with $|\beta|=1$, where $1/2<\gamma\leqslant1$ and $0<\alpha\leqslant\gamma$.\par
$V^{\rm l}\in\mathscr{V}_{\rm G}^{\rm l}$ belongs to $C^2(\mathbb{R}^n)$ and satisfies
\begin{equation}
|\partial_x^\beta V^{\rm l}(x)|\leqslant C_\beta\langle x\rangle^{-\gamma_{\rm G}-\kappa|\beta|}\label{long_range_graf}
\end{equation}
for $|\beta|\leqslant2$, where $0<\gamma_{\rm G}\leqslant1/2$ and $1-\gamma_{\rm G}<\kappa\leqslant1$.\par
Finally, $V^{\rm l}\in\mathscr{V}_{\rm D}^{\rm l}$ belongs to $C^2(\mathbb{R}^n)$ and satisfies
\begin{equation}
|\partial_x^\beta V^{\rm l}(x)|\leqslant C_\beta\langle x\rangle^{-\gamma_{\rm D}-|\beta|/2}\label{long_range_dollard}
\end{equation}
for $|\beta|\leqslant2$, where $3/8<\gamma_{\rm D}\leqslant1/2$.
\end{Ass}

For $V\in\mathscr{V}^{\rm vs}+\mathscr{V}^{\rm s}+(\mathscr{V}_{\rm G}^{\rm l}\cup\mathscr{V}_{\rm D}^{\rm l})$, the full Hamiltonian
\begin{equation}
H^{\rm S}=H_0^{\rm S}+V\label{full}
\end{equation}
is also self-adjoint. $V_1^{\rm vs}$ is unbounded, however, its self-adjointness is established by the Kato--Rellich theorem (see Simon \cite{Si}). We here note that the Coulomb-type local singularity,
\begin{equation}
V_1^{\rm vs}(x)=K|x|^{-1}F(|x|\leqslant1)
\end{equation}
with $0\not=K\in\mathbb{R}$, satisfies Assumption \ref{ass} when $n\geqslant3$.\par
We first consider the short-range case, that is, $V^{\rm l}\equiv0$. The forward two-body short-range scattering in the Stark effect was originally discussed in Avron and Herbst \cite{AvHe}, and Herbst \cite{Her}. Therefore, in this instance, we see that the wave operators defined by the following strong limits
\begin{equation}
W^\pm=\slim_{t\rightarrow\pm\infty}e^{{\rm i}tH^{\rm S}}e^{-{\rm i}tH_0^{\rm S}}\label{wave_operators}
\end{equation}
exist. By using these wave operators $W^\pm$, the scattering operator $S=S(V)$ is defined by
\begin{equation}
S=(W^+)^*W^-.\label{scattering_operator}
\end{equation}
The first theorem of this paper is the following.

\begin{The}\label{the1}
If $S(V_1)=S(V_2)$ for $V_1, V_2\in\mathscr{V}^{\rm vs}+\mathscr{V}^{\rm s}$, then $V_1=V_2$ holds.
\end{The}

If there are no electric fields, that is, for the standard Schr\"odinger operator $|p|^2/2+V$, then $\mathscr{V}^{\rm vs}$ is short-ranged and $\mathscr{V}^{\rm s}$ is long-ranged; otherwise, for the Stark Hamiltonian, Ozawa \cite{Oz} proved that both $\mathscr{V}^{\rm vs}$ and $\mathscr{V}^{\rm s}$ are short-ranged. Theorem \ref{the1} states that, by virtue of the Stark effect, the scattering operator determines the uniqueness of the interaction potentials, which belong to the long-range class in the absence of an electric field.\par
The Enss--Weder time-dependent method was developed in \cite{EnWe} and, by applying its method, Weder \cite{We} first proved this theorem for $\gamma>3/4$. However, as we mentioned above, the borderline between the short-range and long-range is $1/2$. Nicoleau \cite{Ni1} proved this theorem for $V\in C^\infty(\mathbb{R}^n)$ which under $\gamma>1/2$, satisfies
\begin{equation}
|\partial_x^\beta V(x)|\leqslant C_\beta\langle x\rangle^{-\gamma-|\beta|}
\end{equation}
with the additional condition $n\geqslant3$. Thereafter, these results were improved by Adachi--Maehara \cite{AdMa} given $1/2<\alpha\leqslant\gamma$. The behavior of the short-range part under their assumptions was
\begin{equation}
V^{\rm s}(x)=O(|x|^{-1/2-\epsilon}),\quad \nabla_xV^{\rm s}(x)=O(|x|^{-3/2-\epsilon})
\end{equation}
with small $\epsilon>0$. In this sense, a possibility in which the condition regarding the size of $\alpha$ could be relaxed was left because the classical trajectory in the Stark effect is $x(t)=O(t^2)$ as $t\rightarrow\infty$. Recently, Adachi, Fujiwara, and Ishida \cite{AdFuIs} considered the time-dependent electric fields
\begin{equation}
H_0^{\rm S}(t)=|p|^2/2-E(t)\cdot x,\quad E(t)=E_0(1+|t|)^{-\mu},\label{time_dependent}
\end{equation}
where $0\leqslant\mu<1$ and $0\not=E_0\in\mathbb{R}^n$, and proved this theorem under $\tilde{\alpha}_\mu<\alpha\leqslant\gamma$ with $1/(2-\mu)<\gamma\leqslant1$ and
\begin{equation}
\tilde{\alpha}_\mu=
\begin{cases}
\ \displaystyle\frac{7-3\mu-\sqrt{(1-\mu)(17-9\mu)}}{4(2-\mu)} & \mbox{if}\ 0\leqslant\mu\leqslant1/2,\\
\ \displaystyle\frac{1+\mu}{2(2-\mu)} & \mbox{if}\ 1/2<\mu<1.
\end{cases}\label{alpha_mu}
\end{equation}
The smallest $\tilde{\alpha}_\mu$ is when $\mu=0$, and in this case, \eqref{time_dependent} corresponds to the constant electric field \eqref{free}. Therefore, the result by \cite{AdFuIs} is one of the improvements of \cite{AdMa} because
\begin{equation}
\tilde{\alpha}_0=(7-\sqrt{17})/8<1/2.\label{alpha_0}
\end{equation}
Theorem \ref{the1} is a further improvement of \cite{AdMa} and \cite{AdFuIs}. We prove that this $\tilde{\alpha}_0$ is allow to be equal to $0$. This means that the tail of the first-order differential of the short-range part behaves as
\begin{equation}
\nabla_xV^{\rm s}(x)=O(|x|^{-1-\epsilon}).
\end{equation}
Therefore, from the physical aspect and the motion of the classical trajectory, our assumptions are quite natural.\par
We next consider the long-range case, that is, $V^{\rm l}\not=0$. For $V^{\rm l}\in\mathscr{V}_{\rm G}^{\rm l}$, we find the existence of the Graf-type (or Zorbas-type) modified wave operators which were proposed in Zorbas \cite{Zo} and Graf \cite{Gr}
\begin{equation}
W_{\rm G}^{\pm}=\slim_{t\rightarrow\pm\infty}e^{{\rm i}tH^{\rm S}}e^{-{\rm i}tH_0^{\rm S}}e^{-{\rm i}\int_0^tV^{\rm l}(e_1\tau^2/2)d\tau},\label{wave_operators_graf}
\end{equation}
and the Dollard-type modified wave operators introduced by Jensen and Yajima \cite{JeYa} (see also White \cite{Wh} and Adachi \cite{Ad})
\begin{equation}
W_{\rm D}^{\pm}=\slim_{t\rightarrow\pm\infty}e^{{\rm i}tH^{\rm S}}e^{-{\rm i}tH_0^{\rm S}}e^{-{\rm i}\int_0^tV^{\rm l}(p\tau+e_1\tau^2/2)d\tau}\label{wave_operators_dollard}
\end{equation}
by virtue of the condition $\gamma_{\rm G}+\kappa>1$, where $e_1=(1,0,\ldots,0)\in\mathbb{R}^n$. We find also the existence of \eqref{wave_operators_dollard}, even if $V^{\rm l}\in\mathscr{V}_{\rm D}^{\rm l}$. Then, for $V^{\rm l}\in\mathscr{V}_{\rm G}^{\rm l}\cup\mathscr{V}_{\rm D}^{\rm l}$, the Dollard-type modified scattering operator $S_{\rm D}=S_{\rm D}(V^{\rm l};V^{\rm vs}+V^{\rm s})$ is defined by
\begin{equation}
S_{\rm D}=(W_{\rm D}^+)^*W_{\rm D}^-.\label{modified_scattering_operator}
\end{equation}
The second theorem of this paper is the following.

\begin{The}\label{the2}
Let a $V^{\rm l}\in\mathscr{V}_{\rm G}^{\rm l}\cup\mathscr{V}_{\rm D}^{\rm l}$ be given. If $S_{\rm D}(V^{\rm l};V_1)=S_{\rm D}(V^{\rm l};V_2)$ for $V_1, V_2\in\mathscr{V}^{\rm vs}+\mathscr{V}^{\rm s}$, then $V_1=V_2$ holds. Moreover, any one of the Dollard-type modified scattering operators $S_{\rm D}$ determines uniquely the total potential $V$. 
\end{The}

When $V^{\rm l}\in\mathscr{V}_{\rm G}^{\rm l}$, a similar result to Theorem \ref{the2} was obtained in \cite{AdMa} (Note that the notation of $\gamma_{\rm G}$ was denoted by $\gamma_{\rm D}$ in \cite{AdMa}), however, the decay condition of the short-range part was $1/2<\alpha\leqslant\gamma$. Therefore, Theorem \ref{the2} extends the short-range class introduced in \cite{AdMa} to the broader $\mathscr{V}^{\rm s}$. For $V^{\rm l}\in\mathscr{V}_{\rm D}^{\rm l}$, the uniqueness of the short-range interactions was also proved in \cite{AdFuIs} for the time-dependent electric fields \eqref{time_dependent}, in which $\alpha$ satisfied $\tilde{\alpha}_{\mu,\rm D}<\alpha\leqslant\gamma$ with $1/(2-\mu)<\gamma\leqslant1$ and
\begin{equation}
\tilde{\alpha}_{\mu,\rm D}=
\begin{cases}
\ \displaystyle\frac{13-5\mu-\sqrt{(1-\mu)(41-25\mu)}}{8(2-\mu)} & \mbox{if}\ 0\leqslant\mu\leqslant5/7,\\
\ \displaystyle\frac{1+\mu}{2(2-\mu)} & \mbox{if}\ 5/7<\mu<1,
\end{cases}\label{alpha_mu_D}
\end{equation}
and $\gamma_{\rm D}$ satisfied $\tilde{\gamma}_\mu<\gamma_{\rm D}\leqslant1/(2-\mu)$ with
\begin{equation}
\tilde{\gamma}_\mu=\frac{1}{2(2-\mu)}+\frac{1-\mu}{4(2-\mu)}.\label{gamma_mu}
\end{equation}
The smallest $\tilde{\alpha}_{\mu,\rm D}$ and $\tilde{\gamma}_\mu$ are when $\mu=0$, and this case corresponds to a constant electric field \eqref{free}. In comparison with our result, let us substitute $\mu=0$ for \eqref{alpha_mu_D} and \eqref{gamma_mu}. Although $\tilde{\gamma}_0=3/8$ says that the condition on the long-range class is the same as our assumption \eqref{long_range_dollard}, for the short-range class, Theorem \ref{the2} makes true improvement because
\begin{equation}
\tilde{\alpha}_{0,\rm D}=(13-\sqrt{41})/16.\label{alpha_0_D}
\end{equation}
We prove that this $\tilde{\alpha}_{0,\rm D}$ is allow to be equal to $0$.\par
A variety of estimates to prove Theorems \ref{the1} and \ref{the2} shall be given in the following sections. We here emphasize that some of them are new, in particular, as for the long-range case of $V^{\rm l}\in\mathscr{V}_{\rm D}^{\rm l}$, these estimates hold for much broader class $\hat{\mathscr{V}}_{\rm D}^{\rm l}$ defined in Section \ref{long_range_interactions}.\par
There are several other studies concerning the uniqueness of the interaction potentials in the electric fields. Nicoleau \cite{Ni2} considered the time-periodic electric field and obtained the same result given in \cite{Ni1}. Valencia and Weder \cite{VaWe} applied the result obtained in \cite{AdMa} to the $N$-body case (see also \cite{We}). Adachi, Kamada, Kazuno, and Toratani \cite{AdKaKaTo} also treated the time-dependent electric field, which is the same as in \eqref{time_dependent}, however, the case where $\mu=0$, that is, the constant electric field \eqref{free} was not included.

\section{Short-range interactions}\label{short_range_interactions}
In this section, we consider the short-range interactions only, thus $V^{\rm l}\equiv0$, and give a proof of Theorem \ref{the1}. The following reconstruction theorem leads to the proof.

\begin{The}\label{the3}
Let $\omega\in\mathbb{R}^n$ be given such that $|\omega|=1$ and $|\omega\cdot e_1|<1$. Put $v=|v|\omega$. Suppose $\Phi_0,\Psi_0\in L^2(\mathbb{R}^n)$ such that their Fourier transforms $\mathscr{F}\Phi_0,\mathscr{F}\Psi_0\in C_0^\infty(\mathbb{R}^n)$ with $\supp\mathscr{F}\Phi_0,\supp\mathscr{F}\Psi_0\subset\{\xi\in\mathbb{R}^n\bigm||\xi|<\eta\}$ for the given $\eta>0$. Put $\Phi_v=e^{{\rm i}v\cdot x}\Phi_0,\Psi_v=e^{{\rm i}v\cdot x}\Psi_0$. Then
\begin{eqnarray}
|v|({\rm i}[S,p_j]\Phi_v,\Psi_v)\Eqn{=}\int_{-\infty}^\infty\Big\{(V^{\rm vs}(x+\omega t)p_j\Phi_0,\Psi_0)-(V^{\rm vs}(x+\omega t)\Phi_0,p_j\Psi_0)\nonumber \\
\Eqn{}\qquad\quad+({\rm i}(\partial_{x_j}V^{\rm s})(x+\omega t)\Phi_0,\Psi_0)\Big\}dt+o(1)\label{reconstruction}
\end{eqnarray}
holds as $|v|\rightarrow\infty$ for $V^{\rm vs}\in\mathscr{V}^{\rm vs}$ and $V^{\rm s}\in\mathscr{V}^{\rm s}$, where $(\cdot,\cdot)$ is the scalar product of $L^2(\mathbb{R}^n)$ and $p_j$ is the $j$th component of $p$.
\end{The}

In preparation to prove Theorem \ref{the3}, and for use throughout this paper, the following proposition due to Enss \cite[Proposition 2.10]{En}.

\begin{Prop}\label{prop4}
For any $f\in C_0^\infty(\mathbb{R}^n)$ with $\supp f\subset\{\xi\in\mathbb{R}^n\bigm||\xi|<\eta\}$ for some $\eta>0$ and any $N\in\mathbb{N}$, there exists a constant $C_N$ dependent on $f$ only such that
\begin{equation}
\|F(x\in\mathscr{M}')e^{-{\rm i}t|p|^2/2}f(p-v)F(x\in\mathscr{M})\|\leqslant C_N(1+r+|t|)^{-N}
\end{equation}
for $t\in\mathbb{R}$ and measurable sets $\mathscr{M}', \mathscr{M}\subset\mathbb{R}^n$ with the property $r=\normalfont{\textrm{dist}}(\mathscr{M}',\mathscr{M}+vt)-\eta|t|\geqslant0$.
\end{Prop}

We next prepare the following propagation estimate for the singular part $V^{\rm vs}$ which was proved in \cite[Lemma 2.1]{AdMa} (see also \cite{We}). While $\|\cdot\|$ also indicates the norm in $L^2(\mathbb{R}^n)$, for simplicity, we do not distinguish between the notations for the usual $L^2$-norm and its operator norm.

\begin{Prop}\label{prop5}
Let $v$ and $\Phi_v$ be as in Theorem \ref{the3}. Then
\begin{equation}
\int_{-\infty}^\infty\|V^{\rm vs}(x)e^{-{\rm i}tH_0^{\rm S}}\Phi_v\|dt=O(|v|^{-1})
\end{equation}
holds as $|v|\rightarrow\infty$ for $V^{\rm vs}\in\mathscr{V}^{\rm vs}$.
\end{Prop}

The second propagation estimate for the regular part $V^{\rm s}$ is one of the main techniques in this paper, and is also one of the improvements on previous work.

\begin{Prop}\label{prop6}
Let $v$ and $\Phi_v$ be as in Theorem \ref{the3}. Then
\begin{equation}
\int_{-\infty}^\infty\|\{V^{\rm s}(x)-V^{\rm s}(vt+e_1t^2/2)\}e^{-{\rm i}tH_0^{\rm S}}\Phi_v\|dt=O(|v|^{-1})\label{estimate_V^s}
\end{equation}
holds as $|v|\rightarrow\infty$ for $V^{\rm s}\in\mathscr{V}^{\rm s}$.
\end{Prop}

In \cite[Lemma 2.2]{AdMa}, the right-hand side of \eqref{estimate_V^s} was $O(|v|^{-\alpha})$ for $1/2<\alpha<1$. This order was improved in \cite[Lemma 3.4]{AdFuIs} by giving $O(|v|^{\Theta_0(\alpha)+\epsilon})$ with any small $\epsilon>0$ and
\begin{equation}
\Theta_0(\alpha)=-\alpha-\frac{\alpha(1-\alpha)}{2-\alpha}.
\end{equation}
The number $(7-\sqrt{17})/8$ in \eqref{alpha_0} comes from the inequality $\Theta_0(\alpha)<-1/2$, which is required to prove the reconstruction theorem (see \eqref{short_range_remainder} and below in the proof of Theorem \ref{the3}). As mentioned before, not only was the time-independent case \eqref{free} treated by \cite{AdFuIs}, but also the time-dependent case \eqref{time_dependent}. For more details, see \cite[Lemma 3.4]{AdFuIs}. Our key ideas for further improvements are the efficient use of the well-known propagation estimate for the free Schr\"odinger dynamics
\begin{equation}
\|xe^{-{\rm i}t|p|^2/2}\Phi_0\|=O(|t|)\label{free_dynamics1}
\end{equation}
as $|t|\rightarrow\infty$ and the H\"older inequality (see the estimates of $I_3$ and $I_6$ in the proof).

\begin{proof}[Proof of Proposition \ref{prop6}]
Denote the integrand of \eqref{estimate_V^s} by
\begin{equation}
I=\|\{V^{\rm s}(x)-V^{\rm s}(vt+e_1t^2/2)\}e^{-{\rm i}tH_0^{\rm S}}\Phi_v\|.
\end{equation}
The Avron--Herbst formula \cite{AvHe}
\begin{equation}
e^{-{\rm i}tH_0^{\rm S}}=e^{-{\rm i}t^3/6}e^{{\rm i}tx_1}e^{-{\rm i}t^2p_1/2}e^{-{\rm i}t|p|^2/2}\label{avron_herbst}
\end{equation}
connects the propagators of the free Stark Hamiltonian and the standard free Schr\"odinger operator, and plays an important role throughout this paper. Choose $f\in C_0^\infty(\mathbb{R}^n)$ with $\supp f\subset\{\xi\in\mathbb{R}^n\bigm||\xi|<\eta\}$ such that $\mathscr{F}\Phi_0=f\mathscr{F}\Phi_0$ holds. By using \eqref{avron_herbst} and the relation
\begin{equation}
e^{-{\rm i}v\cdot x}e^{-{\rm i}t|p|^2/2}e^{{\rm i}v\cdot x}=e^{-{\rm i}t|v|^2/2}e^{-{\rm i}tp\cdot v}e^{-{\rm i}t|p|^2/2},\label{eq1_1}
\end{equation}
we have
\begin{eqnarray}
I\Eqn{=}\|\{V^{\rm s}(x+e_1t^2/2)-V^{\rm s}(vt+e_1t^2/2)\}e^{-{\rm i}t|p|^2/2}\Phi_v\|\nonumber\\
\Eqn{=}\|\{V^{\rm s}(x+vt+e_1t^2/2)-V^{\rm s}(vt+e_1t^2/2)\}e^{-{\rm i}t|p|^2/2}f(p)\Phi_0\|.\label{eq1_2}
\end{eqnarray}
Split the integral \eqref{estimate_V^s} such that
\begin{equation}
\int_{-\infty}^\infty Idt=\int_{|t|<|v|^{-\sigma_1}}Idt+\int_{|t|\geqslant|v|^{-\sigma_1}}Idt
\end{equation}
with $0<\sigma_1<1$ which is independent of $t$ and $v$. Its lower and upper bounds shall be determined at the end of this proof. We first consider the integral on $|t|<|v|^{-\sigma_1}$. Put $\delta=|\omega\cdot e_1|<1$ and $\lambda=\sqrt{1-\delta^2}/4$. By inserting
\begin{equation}
F(|x|\geqslant3\lambda|v||t|)+F(|x|<3\lambda|v||t|)=1
\end{equation}
between $\{V^{\rm s}(x+vt+e_1t^2/2)-V^{\rm s}(vt+e_1t^2/2)\}$ and $e^{-{\rm i}t|p|^2/2}$, and inserting
\begin{equation}
F(|x|>\lambda|v||t|)+F(|x|\leqslant\lambda|v||t|)=1
\end{equation}
between $f(p)$ and $\Phi_0$ in the term of $F(|x|\geqslant3\lambda|v||t|)$, we find the estimate
\begin{equation}
I\leqslant I_1+I_2+I_3,
\end{equation}
where
\begin{eqnarray}
I_1\Eqn{=}C_1\|F(|x|\geqslant3\lambda|v||t|)e^{-{\rm i}t|p|^2/2}f(p)F(|x|\leqslant\lambda|v||t|)\|,\\
I_2\Eqn{=}C_2\|F(|x|>\lambda|v||t|)\langle x\rangle^{-2}\|,\\
I_3\Eqn{=}\|\{V^{\rm s}(x+vt+e_1t^2/2)-V^{\rm s}(vt+e_1t^2/2)\}F(|x|<3\lambda|v||t|)e^{-{\rm i}t|p|^2/2}\Phi_0\|\hspace{10mm}
\end{eqnarray}
with $C_1=2\|V^{\rm s}\|\|\Phi_0\|$ and $C_2=2\|V^{\rm s}\|\|\langle x\rangle^2\Phi_0\|$. By the relation \eqref{eq1_1} and Proposition \ref{prop4} for $N=2$ with $\lambda|v|\geqslant\eta$, $I_1$ is estimated such that
\begin{equation}
I_1\leqslant C\langle vt\rangle^{-2}\label{eq1_3}.
\end{equation}
Because $I_2$ has this same estimate, it follows that
\begin{equation}
\int_{|t|<|v|^{-\sigma_1}}(I_1+I_2)dt\leqslant C\int_0^{|v|^{-\sigma_1}}\langle vt\rangle^{-2}dt=C|v|^{-1}\int_0^{|v|^{1-\sigma_1}}\langle\tau\rangle^{-2}d\tau=O(|v|^{-1}).\label{eq1_4}
\end{equation}
For $I_3$, we compute
\begin{eqnarray}
I_3\Eqn{=}\|\int_0^1(\nabla_x V^{\rm s})(\theta x+vt+e_1t^2/2)\cdot xd\theta F(|x|<3\lambda|v||t|)e^{-{\rm i}t|p|^2/2}\Phi_0\|\nonumber\\
\Eqn{\leqslant}\int_0^1\|(\nabla_x V^{\rm s})(\theta x+vt+e_1t^2/2)F(|x|<3\lambda|v||t|)\|d\theta\times\|xe^{-{\rm i}t|p|^2/2}\Phi_0\|.\label{eq1_5}\qquad
\end{eqnarray}
In the same manner as for the proof of \cite[Lemma 2.1]{AdMa}, we have
\begin{gather}
|vt+e_1t^2/2|^2=|v|^2t^2+t^4/4+v\cdot e_1t^3\geqslant|v|^2t^2+t^4/4-\delta|v||t|^3\nonumber\\
=t^2(|t|-2\delta|v|)^2/4+(1-\delta^2)|v|^2t^2\geqslant(1-\delta^2)|v|^2t^2.\label{eq1_6}
\end{gather}
and $|vt+e_1t^2/2|\geqslant4\lambda|v||t|$. Therefore, when $|x|<3\lambda|v||t|$ holds, we obtain
\begin{equation}
|\theta x+vt+e_1t^2/2|\geqslant|vt+e_1t^2/2|-|x|>\lambda|v||t|.\label{eq1_7}
\end{equation}
We thus estimate, by this inequality \eqref{eq1_7} and the short-range assumption \eqref{short_range},
\begin{equation}
\int_0^1\|(\nabla_x V^{\rm s})(\theta x+vt+e_1t^2/2)F(|x|<3\lambda|v||t|)\|d\theta\leqslant C\langle vt\rangle^{-1-\alpha}.\label{eq1_8}
\end{equation}
From the Heisenberg picture of $x$,
\begin{equation}
e^{{\rm i}t|p|^2/2}xe^{-{\rm i}t|p|^2/2}=x+tp,\label{heisenberg_x}
\end{equation}
we have
\begin{equation}
\|xe^{-{\rm i}t|p|^2/2}\Phi_0\|\leqslant C(1+|t|)\label{free_dynamics2}
\end{equation}
which holds for all $t\in\mathbb{R}$. From the estimates \eqref{eq1_8} and \eqref{free_dynamics2}, $I_3$ is
\begin{equation}
I_3\leqslant C\langle vt\rangle^{-1-\alpha}(1+|t|).\label{eq1_9}
\end{equation}
We compute, for $\alpha<1$,
\begin{gather}
\int_0^{|v|^{-\sigma_1}}\langle vt\rangle^{-1-\alpha}tdt=|v|^{-1-\alpha}\int_0^{|v|^{-\sigma_1}}t^{-\alpha}dt=O(|v|^{-1-\alpha})\label{eq1_10},\\
\int_0^{|v|^{-\sigma_1}}\langle vt\rangle^{-1-\alpha}dt=|v|^{-1}\int_0^{|v|^{1-\sigma_1}}\langle\tau\rangle^{-1-\alpha}d\tau=O(|v|^{-1}).\label{eq1_11}
\end{gather}
If $\alpha=1$, \eqref{eq1_10} reduces to the estimate
\begin{eqnarray}
\lefteqn{\int_0^{|v|^{-\sigma_1}}\langle vt\rangle^{-2}tdt=|v|^{-2}\int_0^{|v|^{1-\sigma_1}}\langle\tau\rangle^{-2}\tau d\tau}\nonumber\\
\Eqn{}\leqslant|v|^{-2}\int_0^{|v|^{1-\sigma_1}}\langle\tau\rangle^{-1}d\tau=O(|v|^{-2}\log|v|).
\end{eqnarray}
However, we can assume $\alpha<1$ without loss of generality. \eqref{eq1_10} and \eqref{eq1_11} state that
\begin{equation}
\int_{|t|<|v|^{-\sigma_1}}I_3dt=O(|v|^{-1}).\label{eq1_12}
\end{equation}
We next consider the integral on $|t|\geqslant|v|^{-\sigma_1}$. Take $0<\sigma_2<1$ which is also independent of $t$ and $v$. Its lower and upper bounds shall be determined at the end of this proof. By replacing $\lambda|v|$ by $|v|^{\sigma_2}$ in $I_1$, $I_2$, and $I_3$, we have
\begin{equation}
I\leqslant I_4+I_5+I_6,
\end{equation}
where
\begin{eqnarray}
I_4\Eqn{=}C_1\|F(|x|\geqslant3|v|^{\sigma_2}|t|)e^{-{\rm i}t|p|^2/2}f(p)F(|x|\leqslant|v|^{\sigma_2}|t|)\|,\\
I_5\Eqn{=}C_2\|F(|x|>|v|^{\sigma_2}|t|)\langle x\rangle^{-2}\|,\\
I_6\Eqn{=}\|\{V^{\rm s}(x+vt+e_1t^2/2)-V^{\rm s}(vt+e_1t^2/2)\}\nonumber\\
\Eqn{}\hspace{30mm}\times F(|x|<3|v|^{\sigma_2}|t|)e^{-{\rm i}t|p|^2/2}\Phi_0\|.
\end{eqnarray}
By Proposition \ref{prop4} for $N=2$ with $|v|^{\sigma_2}\geqslant\eta$, $I_4$ is estimated such that
\begin{equation}
I_4\leqslant C\langle |v|^{\sigma_2}t\rangle^{-2}\label{eq1_13}.
\end{equation}
It therefore follows that
\begin{equation}
\int_{|t|\geqslant|v|^{-\sigma_1}}(I_4+I_5)dt\leqslant C\int_{|v|^{-\sigma_1}}^\infty\langle |v|^{\sigma_2}t\rangle^{-2}dt=C|v|^{-2\sigma_2}\int_{|v|^{-\sigma_1}}^\infty t^{-2}dt=O(|v|^{\sigma_1-2\sigma_2}).\label{eq1_14}
\end{equation}
Similar to the proof of \cite[Lemma 2.2]{AdMa}, when $|x|<3|v|^{\sigma_2}|t|$ and $3|v|^{\sigma_2-1}\leqslant(1-\delta)/4$,
\begin{eqnarray}
\lefteqn{|x+vt+e_1t^2/2|^2=|x+vt|^2+t^4/4+(x+vt)\cdot e_1t^2}\nonumber\\
\Eqn{}>(1-3|v|^{\sigma_2-1})^2|v|^2t^2+t^4/4-(\delta+3|v|^{\sigma_2-1})|v||t|^3\nonumber\\
\Eqn{}\geqslant(3+\delta)^2|v|^2t^2/16+t^4/4-(1+3\delta)|v||t|^3/4\label{eq1_15}
\end{eqnarray}
holds. We thus obtain
\begin{equation}
|x+vt+e_1t^2/2|^2\geqslant\{(3+\delta)^2-(1+3\delta)^2\}|v|^2t^2/16=(1-\delta^2)|v|^2t^2/2\label{eq1_16}
\end{equation}
because
\begin{equation}
t^4/4-(1+3\delta)|v||t|^3/4+(1+3\delta)^2|v|^2t^2/16=\{t-(1+3\delta)/2\}^2t^2/4\geqslant0.
\end{equation}
Moreover, we also obtain
\begin{equation}
|x+vt+e_1t^2/2|^2\geqslant\{1-(1+3\delta)^2/(3+\delta)^2\}t^4/4=2(1-\delta^2)t^4/(3+\delta)^2\label{eq1_17}
\end{equation}
because
\begin{gather}
(3+\delta)^2|v|^2t^2/16-(1+3\delta)|v||t|^3/4+(1+3\delta)^2/(3+\delta)^2t^4/4\nonumber\\
=\{(1+3\delta)t/(3+\delta)-(3+\delta)|v|/2\}^2t^2/4\geqslant0.
\end{gather}
Therefore, from \eqref{eq1_16} and \eqref{eq1_17}, if $|x|<3|v|^{\sigma_2}|t|$ and $3|v|^{\sigma_2-1}\leqslant(1-\delta)/4$, we have
\begin{equation}
|x+vt+e_1t^2/2|\geqslant\max\{c_1|v||t|, c_2t^2 \}\label{eq1_18}
\end{equation}
with $c_1=\sqrt{(1-\delta^2)/2}$ and $c_2=\sqrt{2(1-\delta^2)}/(3+\delta)$. As in \eqref{eq1_5}, $I_6$ is
\begin{equation}
I_6\leqslant\int_0^1\|(\nabla_x V^{\rm s})(\theta x+vt+e_1t^2/2)F(|x|<3|v|^{\sigma_2}|t|)\|d\theta\times\|xe^{-{\rm i}t|p|^2/2}\Phi_0\|.\label{eq1_19}
\end{equation}
By using the H\"older inequality, we estimate the integral of $I_6$ such that
\begin{equation}
\int_{|v|^{-\sigma_1}}^\infty I_6dt\leqslant\left(\int_{|v|^{-\sigma_1}}^\infty I_6^{q_1}dt\right)^{1/q_1}\left(\int_{|v|^{-\sigma_1}}^\infty I_6^{q_2}dt\right)^{1/q_2},\label{eq1_20}
\end{equation}
where $q_1$ and $q_2$ are the H\"older conjugates of each other, that is, $1/q_1+1/q_2=1$ for $q_1>1$. From the estimate
\begin{equation}
|x+vt+e_1t^2/2|\geqslant c_1|v||t|
\end{equation}
of \eqref{eq1_18} and assumption \eqref{short_range}, we have
\begin{eqnarray}
\lefteqn{\int_{|v|^{-\sigma_1}}^\infty I_6^{q_1}dt\leqslant C\int_{|v|^{-\sigma_1}}^\infty\{\langle vt\rangle^{-1-\alpha}(1+t)\}^{q_1}dt}\nonumber\\
\Eqn{}\leqslant 2^{q_1-1}C\left(\int_{|v|^{-\sigma_1}}^\infty\langle vt\rangle^{-q_1(1+\alpha)}dt+\int_{|v|^{-\sigma_1}}^\infty\langle vt\rangle^{-q_1(1+\alpha)}t^{q_1}dt\right).\label{eq1_21}
\end{eqnarray}
Choose $q_1$ which satisfies $q_1\alpha>1$, and we then compute
\begin{equation}
\int_{|v|^{-\sigma_1}}^\infty\langle vt\rangle^{-q_1(1+\alpha)}t^{q_1}dt\leqslant|v|^{-q_1(1+\alpha)}\int_{|v|^{-\sigma_1}}^\infty t^{-q_1\alpha}dt=O(|v|^{-q_1(1+\alpha)-\sigma_1(1-q_1\alpha)}).\label{eq1_22}
\end{equation}
We also compute
\begin{equation}
\int_{|v|^{-\sigma_1}}^\infty\langle vt\rangle^{-q_1(1+\alpha)}dt\leqslant|v|^{-q_1(1+\alpha)}\int_{|v|^{-\sigma_1}}^\infty t^{-q_1(1+\alpha)}dt=O(|v|^{-q_1(1+\alpha)-\sigma_1\{1-q_1(1+\alpha)\}}).\label{eq1_23}
\end{equation}
\eqref{eq1_22} and \eqref{eq1_23} state that
\begin{equation}
\int_{|v|^{-\sigma_1}}^\infty I_6^{q_1}dt=O(|v|^{-q_1(1+\alpha)-\sigma_1\{1-q_1(1+\alpha)\}}).\label{eq1_24}
\end{equation}
Again, from
\begin{equation}
|x+vt+e_1t^2/2|\geqslant c_2t^2
\end{equation}
of \eqref{eq1_18} and \eqref{short_range}, we compute
\begin{eqnarray}
\lefteqn{\int_{|v|^{-\sigma_1}}^\infty I_6^{q_2}dt\leqslant C\int_{|v|^{-\sigma_1}}^\infty\{\langle t^2\rangle^{-1-\alpha}(1+t)\}^{q_2}dt}\nonumber\\
\Eqn{}\leqslant 2^{q_2-1}C\left(\int_{|v|^{-\sigma_1}}^\infty\langle t^2\rangle^{-q_2(1+\alpha)}dt+\int_{|v|^{-\sigma_1}}^\infty\langle t^2\rangle^{-q_2(1+\alpha)}t^{q_2}dt\right)\nonumber\\
\Eqn{}=O(|v|^{-\sigma_1\{1-2q_2(1+\alpha)\}})+O(|v|^{-\sigma_1\{1-q_2(1+2\alpha)\}})=O(|v|^{-\sigma_1\{1-2q_2(1+\alpha)\}}).\qquad\label{eq1_25}
\end{eqnarray}
We here used $2q_2(1+\alpha)>q_2(1+2\alpha)>1$. By combining \eqref{eq1_20}, \eqref{eq1_24}, \eqref{eq1_25}, and
\begin{equation}
-1-\alpha-\sigma_1\{1-q_1(1+\alpha)\}/q_1-\sigma_1\{1-2q_2(1+\alpha)\}/q_2=-1-\alpha+\sigma_1(2+3\alpha),
\end{equation}
we obtain
\begin{equation}
\int_{|t|\geqslant|v|^{-\sigma_1}}I_6dt=O(|v|^{-1-\alpha+\sigma_1(2+3\alpha)}).\label{eq1_26}
\end{equation}
Together, \eqref{eq1_4}, \eqref{eq1_12}, \eqref{eq1_14}, and \eqref{eq1_26} imply that
\begin{equation}
\int_{-\infty}^\infty Idt=O(|v|^{-1})+O(|v|^{\sigma_1-2\sigma_2})+O(|v|^{-1-\alpha+\sigma_1(2+3\alpha)}).
\end{equation}
To complete our proof, it is sufficient to determine the size of $\sigma_1$ and $\sigma_2$ such that $0<\sigma_1<\alpha/(2+3\alpha)$ and $(1+\sigma_1)/2<\sigma_2<1$. Indeed, $\sigma_1<\alpha/(2+3\alpha)$ is equivalent to $-1-\alpha+\sigma_1(2+3\alpha)<-1$, and $(1+\sigma_1)/2<\sigma_2$ is equivalent to $\sigma_1-2\sigma_2<-1$.
\end{proof}

We introduce the auxiliary Graf-type modified wave operators
\begin{equation}
\Omega_{{\rm G},v}^\pm=\slim_{t\rightarrow\pm\infty}e^{{\rm i}tH^{\rm S}}U_{{\rm G},v}(t),\quad U_{{\rm G},v}(t)=e^{-{\rm i}tH_0^{\rm S}}M_{{\rm G},v}(t)
\end{equation}
with
\begin{equation}
M_{{\rm G},v}(t)=e^{-{\rm i}\int_0^tV^{\rm s}(v\tau+e_1\tau^2/2)d\tau},\label{garf_modifier}
\end{equation}
according to \cite{AdMa} (see also \cite{AdKaKaTo} and \cite{AdFuIs}). Recall \eqref{eq1_6}. Then, by the estimate
\begin{equation}
|vt+e_1t^2/2|^2\geqslant t^2(|v|-\delta|t|/2)^2+(1-\delta^2)t^4/4\geqslant(1-\delta^2)t^4/4
\end{equation}
and assumption $\gamma>1/2$, we see that
\begin{equation}
I_{{\rm G},v}^\pm(t)=\lim_{t\rightarrow\pm\infty}M_{{\rm G},v}(t)\label{I_G_v}
\end{equation}
exist. Because the wave operators \eqref{wave_operators} and this limits \eqref{I_G_v} exist, we also see that $\Omega_{{\rm G},v}^\pm=W^\pm I_{{\rm G},v}^\pm(t)$. \cite{Ni1, Ni2} applied the Dollard-type modification to short-range inverse scattering in the Stark effect. The Graf-type modification was first introduced in \cite{AdMa}. We emphasize that the Graf-type modifier $M_{{\rm G},v}(t)$ is scalar-valued and therefore commutes with any operators.\par

By virtue of Propositions \ref{prop5} and \ref{prop6}, the following corollary is proved as in \cite[Lemma 2.3]{AdMa} and \cite[Lemma 3.5]{AdKaKaTo}. We therefore omit its proof.

\begin{Cor}\label{cor7}
Let $v$ and $\Phi_v$ be as in Theorem \ref{the3}. Then
\begin{equation}
\|\{e^{-{\rm i}tH^{\rm S}}\Omega_{{\rm G},v}^\pm-U_{{\rm G},v}(t)\}\Phi_v\|=O(|v|^{-1})
\end{equation}
holds as $|v|\rightarrow\infty$ uniformly in $t\in\mathbb{R}$, for $V^{\rm vs}\in\mathscr{V}^{\rm vs}$ and $V^{\rm s}\in\mathscr{V}^{\rm s}$.
\end{Cor}

We now prove Theorem \ref{the3}.
\begin{proof}[Proof of Theorem \ref{the3}]
This proof follows similarly \cite[Theorem 2.4]{We} and \cite[Theorem 3.1]{AdMa} (see also \cite{AdKaKaTo} and \cite{AdFuIs}). We give here a sketch. Note that the scattering operator $S$ is represented such that
\begin{equation}
S=I_{{\rm G},v}(\Omega_{{\rm G},v}^+)^*\Omega_{{\rm G},v}^-,\quad I_{{\rm G},v}=I_{{\rm G},v}^+\overline{I_{{\rm G},v}^-}=e^{-{\rm i}\int_{-\infty}^\infty V^{\rm s}(v\tau+e_1\tau^2/2)d\tau}
\end{equation}
and put $V_{v,t}^{\rm s}$ by
\begin{equation}
V_{v,t}^{\rm s}=V^{\rm vs}(x)+V^{\rm s}(x)-V^{\rm s}(vt+e_1t^2/2),
\end{equation}
then we have
\begin{gather}
{\rm i}(S-I_{{\rm G},v})\Phi_v={\rm i}I_{{\rm G},v}(\Omega_{{\rm G},v}^+-\Omega_{{\rm G},v}^-)^*\Omega_{{\rm G},v}^-\Phi_v\nonumber\\
=I_{{\rm G},v}\int_{-\infty}^\infty U_{{\rm G},v}(t)^*V_{v,t}^{\rm s}e^{-{\rm i}tH^{\rm S}}\Omega_{{\rm G},v}^-\Phi_vdt.
\end{gather}
We therefore obtain, by using the relations $[S, p_j]=[S-I_{{\rm G},v},p_j-v_j]$ and $(p_j-v_j)\Phi_v=(p_j\Phi_0)_v$,
\begin{equation}
|v|({\rm i}[S,p_j]\Phi_v,\Psi_v)=I_{{\rm G},v}\{I(v)+R(v)\}
\end{equation}
with
\begin{eqnarray}
I(v)\Eqn{=}|v|\int_{-\infty}^\infty\Big\{(V_{v,t}^{\rm s}U_{{\rm G},v}(t)(p_j\Phi_0)_v,U_{{\rm G},v}(t)\Psi_v)\nonumber\\
\Eqn{}\hspace{20mm}-(V_{v,t}^{\rm s}U_{{\rm G},v}(t)\Phi_v,U_{{\rm G},v}(t)(p_j\Psi_0)_v)\Big\}dt,\\
R(v)\Eqn{=}|v|\int_{-\infty}^\infty(\{e^{-{\rm i}tH^{\rm S}}\Omega_{{\rm G},v}^--U_{{\rm G},v}(t)\}(p_j\Phi_0)_v,V_{v,t}^{\rm s}U_{{\rm G},v}(t)\Psi_v)dt\nonumber\\
\Eqn{}-|v|\int_{-\infty}^\infty(\{e^{-{\rm i}tH^{\rm S}}\Omega_{{\rm G},v}^--U_{{\rm G},v}(t)\}\Phi_v,V_{v,t}^{\rm s}U_{{\rm G},v}(t)(p_j\Psi_0)_v)dt.\quad
\end{eqnarray}
Propositions \ref{prop5}, \ref{prop6} and Corollary \ref{cor7} immediately imply that
\begin{equation}
R(v)=O(|v|^{-1}).\label{short_range_remainder}
\end{equation}
The corresponding result to \eqref{short_range_remainder} in \cite{AdMa} was $O(|v|^{1-2\alpha})$ if $\alpha<1$, and $1-2\alpha<0$ required $\alpha>1/2$. In \cite{AdFuIs}, $R(v)$ was estimated by $O(|v|^{1+2\{\Theta_0(\alpha)+\epsilon\}})$ for any small $\epsilon>0$, and $1+2\{\Theta_0(\alpha)+\epsilon\}<0$ required $\alpha>(7-\sqrt{17})/8$ when $\mu=0$. With our new result, we prove Proposition \ref{prop6} independently of $\alpha$, therefore \eqref{short_range_remainder} is obtained under the weak condition $\alpha>0$.\par
Because the rest of this proof is the same as in \cite{We} and \cite{AdMa}, we omit it here.
\end{proof}

By the Plancherel formula as applied to the Radon transform (see Helgason \cite[Theorem 2.17 in Chap. I]{Hel}), Theorem \ref{the2} is proved similarly to \cite[Theorem 1.2 ]{We} (see also \cite{EnWe}). We have therefore omitted its proof.

\section{Long-range interactions}\label{long_range_interactions}
Our main purpose in this section is proving the following reconstruction theorem. This theorem yields the proof of Theorem \ref{the2}.

\begin{The}\label{the8}
Let $\omega\in\mathbb{R}^n$ be given such that $|\omega|=1$ and $|\omega\cdot e_1|<1$. Put $v=|v|\omega$. Suppose $\Phi_0,\Psi_0\in L^2(\mathbb{R}^n)$ such that $\mathscr{F}\Phi_0,\mathscr{F}\Psi_0\in C_0^\infty(\mathbb{R}^n)$ with $\supp\mathscr{F}\Phi_0,\supp\mathscr{F}\Psi_0\subset\{\xi\in\mathbb{R}^n\bigm||\xi|<\eta\}$ for the given $\eta>0$. Put $\Phi_v=e^{{\rm i}v\cdot x}\Phi_0,\Psi_v=e^{{\rm i}v\cdot x}\Psi_0$. Then
\begin{eqnarray}
\lefteqn{|v|({\rm i}[S_{\rm D},p_j]\Phi_v,\Psi_v)}\nonumber\\
\Eqn{=}\int_{-\infty}^\infty\Big\{(V^{\rm vs}(x+\omega t)p_j\Phi_0,\Psi_0)-(V^{\rm vs}(x+\omega t)\Phi_0,p_j\Psi_0)\nonumber \\
\Eqn{}\qquad+({\rm i}(\partial_{x_j}V^{\rm s})(x+\omega t)\Phi_0,\Psi_0)+({\rm i}(\partial_{x_j}V^{\rm l})(x+\omega t)\Phi_0,\Psi_0)\Big\}dt+o(1)\quad\label{reconstruction_long_range}
\end{eqnarray}
holds as $|v|\rightarrow\infty$ for $V^{\rm vs}\in\mathscr{V}^{\rm vs}$, $V^{\rm s}\in\mathscr{V}^{\rm s}$, and $V^{\rm l}\in\mathscr{V}_{\rm G}^{\rm l}\cup\mathscr{V}_{\rm D}^{\rm l}$.
\end{The}

Before the proof of Theorem \ref{the8}, we prepare some Lemmas and Propositions. To begin, we define a class of long-range potentials $\hat{\mathscr{V}}_{\rm D}^{\rm l}$ as follows. $V^{\rm l}\in\hat{\mathscr{V}}_{\rm D}^{\rm l}$ belongs to $C^2(\mathbb{R}^n)$ and satisfies that
\begin{equation}
|\partial_x^\beta V^{\rm l}(x)|\leqslant C_\beta\langle x\rangle^{-\hat{\gamma}_{\rm D}-|\beta|/2}\label{long_range_dollard_wide}
\end{equation}
for $|\beta|\leqslant2$, where $1/4<\hat{\gamma}_{\rm D}\leqslant1/2$. Clearly, $\mathscr{V}_{\rm D}^{\rm l}\subsetneq\hat{\mathscr{V}}_{\rm D}^{\rm l}$. Moreover, we denote the Dollard-type modifiers $M_{\rm D}(t)$ and $M_{{\rm D}, v}(t)$ by
\begin{gather}
M_{\rm D}(t)=e^{-{\rm i}\int_0^tV^{\rm l}(p\tau+e_1\tau^2/2)d\tau},\label{dollard_modifier}\\
M_{{\rm D}, v}(t)=e^{-{\rm i}v\cdot x}M_{\rm D}(t)e^{{\rm i}v\cdot x}=e^{-{\rm i}\int_0^tV^{\rm l}(p\tau+v\tau+e_1\tau^2/2)d\tau}\label{dollard_modifier_v}
\end{gather}
 for $V^{\rm l}\in\mathscr{V}_{\rm G}^{\rm l}\cup\hat{\mathscr{V}}_{\rm D}^{\rm l}$. \par
 
We first give the estimates when $V^{\rm l}$ belongs to $\mathscr{V}_{\rm G}^{\rm l}$. Lemma \ref{lem9} of which $k=2$ was obtained in \cite[Lemma 3.1]{AdMa}. The proof of $k=1$ was included in the proof of $k=2$. Propositions \ref{prop10} and \ref{prop12} below were also proved in \cite[Lemmas 3.2 and 3.4]{AdMa}.
 
\begin{Lem}\label{lem9}
Let $v$ and $\Phi_v$ be as in Theorem \ref{the8}. Then, for $V^{\rm l}\in\mathscr{V}_{\rm G}^{\rm l}$,
\begin{equation}
\|\langle x\rangle^kM_{{\rm D}, v}(t)\Phi_0\|=O(1)\label{eq9_1}
\end{equation}
with $k=1,2$ holds as $|v|\rightarrow\infty$ uniformly in $t\in\mathbb{R}$.
\end{Lem}

\begin{Prop}\label{prop10}
Let $v$ and $\Phi_v$ be as in Theorem \ref{the8}. Then
\begin{equation}
\int_{-\infty}^\infty\|V^{\rm vs}(x)e^{-{\rm i}tH_0^{\rm S}}M_{\rm D}(t)\Phi_v\|dt=O(|v|^{-1})
\end{equation}
holds as $|v|\rightarrow\infty$ for $V^{\rm vs}\in\mathscr{V}^{\rm vs}$ and $V^{\rm l}\in\mathscr{V}_{\rm G}^{\rm l}$.
\end{Prop}

The next propagation estimate for the regular short-range part $V^{\rm s}$ along the modified time evolution by $M_{\rm D}(t)e^{-{\rm i}H_0^{\rm S}}$ is also one of the new results in this paper. In \cite[Lemma 3.3]{AdMa}, the right-hand side of \eqref{estimate_V^s_long_range_graf} was given as $O(|v|^{-\alpha})$ for $1/2<\alpha<1$.
\begin{Prop}\label{prop11}
Let $v$ and $\Phi_v$ be as in Theorem \ref{the8}. Then
\begin{equation}
\int_{-\infty}^\infty\|\{V^{\rm s}(x)-V^{\rm s}(vt+e_1t^2/2)\}e^{-{\rm i}tH_0^{\rm S}}M_{\rm D}(t)\Phi_v\|dt=O(|v|^{-1})\label{estimate_V^s_long_range_graf}
\end{equation}
holds as $|v|\rightarrow\infty$ for $V^{\rm s}\in\mathscr{V}^{\rm s}$ and $V^{\rm l}\in\mathscr{V}_{\rm G}^{\rm l}$.
\end{Prop}

\begin{proof}[Proof of Proposition \ref{prop11}]
It follows from the Avron--Herbst formula \eqref{avron_herbst} and relation \eqref{eq1_1} that the integrand of \eqref{estimate_V^s_long_range_graf} is
\begin{eqnarray}
\lefteqn{\|\{V^{\rm s}(x)-V^{\rm s}(vt+e_1t^2/2)\}e^{-{\rm i}tH_0^{\rm S}}M_{\rm D}(t)\Phi_v\|}\nonumber\\
\Eqn{}=\|\{V^{\rm s}(x+vt+e_1t^2/2)-V^{\rm s}(vt+e_1t^2/2)\}e^{-{\rm i}t|p|^2/2}M_{{\rm D}, v}(t)f(p)\Phi_0\|\label{eq11_1}
\end{eqnarray}
Therefore, by virtue of Lemma \ref{lem9}, this proof is performed almost in a similar manner to the proof of Proposition \ref{prop6}. We here only note that
\begin{equation}
\|xe^{-{\rm i}t|p|^2/2}M_{{\rm D}, v}(t)\Phi_0\|=\|(x+pt)M_{{\rm D}, v}(t)\Phi_0\|\leqslant C(1+|t|)\label{free_dynamics3}
\end{equation}
holds by \eqref{heisenberg_x} and \eqref{eq9_1} with $k=1$.
\end{proof}

As we stated above, the following propagation estimate for the long-range part $V^{\rm l}\in\mathscr{V}_{\rm G}^{\rm l}$ was already obtained in \cite{AdMa}.
\begin{Prop}\label{prop12}
Let $v$ and $\Phi_v$ be as in Theorem \ref{the8}. Then
\begin{equation}
\int_{-\infty}^\infty\|\{V^{\rm l}(x)-V^{\rm l}(pt-e_1t^2/2)\}e^{-{\rm i}tH_0^{\rm S}}M_{\rm D}(t)\Phi_v\|dt=O(|v|^{-1/2-\epsilon_1})\label{estimate_V^l_long_range_graf}
\end{equation}
holds with some $\epsilon_1>0$ as $|v|\rightarrow\infty$ for $V^{\rm l}\in\mathscr{V}_{\rm G}^{\rm l}$.
\end{Prop}

The following Lemma and Propositions are the estimates when $V^{\rm l}$ belongs to $\hat{\mathscr{V}}_{\rm D}^{\rm l}$. \eqref{eq13_2} of Lemma \ref{lem13} is one of the simple versions in \cite[Lemma 4.2]{AdFuIs}. The proof of \eqref{eq13_1} was included in the proof of \eqref{eq13_2}. Proposition \ref{prop14} was also proved in \cite[Lemma 4.3]{AdFuIs}.
\begin{Lem}\label{lem13}
Let $v$ and $\Phi_v$ be as in Theorem \ref{the8}. Then, for $V^{\rm l}\in\hat{\mathscr{V}}_{\rm D}^{\rm l}$, there exists a positive constant $C$ which is independent of $t$ and $v$ such that
\begin{gather}
\|\langle x\rangle M_{{\rm D}, v}(t)\Phi_0\|\leqslant C(1+|t|^{1-2\hat{\gamma}_{\rm D}}),\label{eq13_1}\\
\|\langle x\rangle^2M_{{\rm D}, v}(t)\Phi_0\|\leqslant C(1+|t|^{1-2\hat{\gamma}_{\rm D}}+|t|^{2-4\hat{\gamma}_{\rm D}}).\label{eq13_2}
\end{gather}
\end{Lem}

\begin{Prop}\label{prop14}
Let $v$ and $\Phi_v$ be as in Theorem \ref{the8}. Then
\begin{equation}
\int_{-\infty}^\infty\|V^{\rm vs}(x)e^{-{\rm i}tH_0^{\rm S}}M_{\rm D}(t)\Phi_v\|dt=O(|v|^{-1})
\end{equation}
holds as $|v|\rightarrow\infty$ for $V^{\rm vs}\in\mathscr{V}^{\rm vs}$ and $V^{\rm l}\in\hat{\mathscr{V}}_{\rm D}^{\rm l}$.
\end{Prop}

The next propagation estimate for $V^{\rm s}$ along the modified time evolution by $M_{\rm D}(t)e^{-{\rm i}H_0^{\rm S}}$ when $V^{\rm l}\in\hat{\mathscr{V}}_{\rm D}^{\rm l}$ is one of the main techniques in this paper, and is also one of the improvements on the previous work.
\begin{Prop}\label{prop15}
Let $v$ and $\Phi_v$ be as in Theorem \ref{the8}. Then
\begin{equation}
\int_{-\infty}^\infty\|\{V^{\rm s}(x)-V^{\rm s}(vt+e_1t^2/2)\}e^{-{\rm i}tH_0^{\rm S}}M_{\rm D}(t)\Phi_v\|dt=O(|v|^{-1})\label{estimate_V^s_long_range_dollard}
\end{equation}
holds as $|v|\rightarrow\infty$ for $V^{\rm s}\in\mathscr{V}^{\rm s}$ and $V^{\rm l}\in\hat{\mathscr{V}}_{\rm D}^{\rm l}$.
\end{Prop}

In \cite[Lemma 4.4]{AdFuIs}, when $\mu=0$ of \eqref{time_dependent}, the estimate of \eqref{estimate_V^s_long_range_dollard} was $O(|v|^{\Theta_{0, {\rm D}}(\alpha)+\epsilon})$ with any small $\epsilon>0$ and
\begin{equation}
\Theta_{0, {\rm D}}(\alpha)=-\alpha-\frac{\alpha(1-\alpha)}{4-3\alpha}.
\end{equation}
The number $(13-\sqrt{41})/16$ in \eqref{alpha_0_D} comes from the inequality $\Theta_{0, {\rm D}}(\alpha)<-1/2$. Our key ideas for this improvement are the efficient use of the propagation estimate of the free Schr\"odinger dynamics and the H\"older inequality as with the proof of Proposition \ref{prop6}.

\begin{proof}[Proof of Proposition \ref{prop15}]
By the formulae \eqref{avron_herbst} and \eqref{eq1_1}, the integrand of \eqref{estimate_V^s_long_range_dollard} is
\begin{equation}
I=\|\{V^{\rm s}(x+vt+e_1t^2/2)-V^{\rm s}(vt+e_1t^2/2)\}e^{-{\rm i}t|p|^2/2}M_{{\rm D}, v}(t)f(p)\Phi_0\|.
\end{equation}
Split the integral \eqref{estimate_V^s_long_range_dollard} such that
\begin{equation}
\int_{-\infty}^\infty Idt=\int_{|t|<|v|^{-\sigma_1}}Idt+\int_{|t|\geqslant|v|^{-\sigma_1}}Idt
\end{equation}
with $0<\sigma_1<1$ which is independent of $t$ and $v$. Its lower and upper bounds are shall be determined at the end of this proof. We first consider the integral on $|t|<|v|^{-\sigma_1}$. Similar to the proof of Proposition \ref{prop6}, put $\lambda=\sqrt{1-\delta^2}/4$ for $\delta=|\omega\cdot e_1|<1$, and we then estimate
\begin{equation}
I\leqslant I_1+I_2+I_3,
\end{equation}
where
\begin{eqnarray}
I_1\Eqn{=}C_1\|F(|x|\geqslant3\lambda|v||t|)e^{-{\rm i}t|p|^2/2}f(p)F(|x|\leqslant\lambda|v||t|)\|,\\
I_2\Eqn{=}C_2\|F(|x|>\lambda|v||t|)\langle x\rangle^{-2}\|\|\langle x\rangle^2M_{{\rm D}, v}(t)\Phi_0\|,\\
I_3\Eqn{=}\|\{V^{\rm s}(x+vt+e_1t^2/2)-V^{\rm s}(vt+e_1t^2/2)\}\nonumber\\
\Eqn{}\hspace{20mm}\times F(|x|<3\lambda|v||t|)e^{-{\rm i}t|p|^2/2}M_{{\rm D}, v}(t)\Phi_0\|
\end{eqnarray}
with $C_1=2\|V^{\rm s}\|\|\Phi_0\|$ and $C_2=2\|V^{\rm s}\|$. $I_1$ has the same shape in the proof of Proposition \ref{prop6}. Therefore, by using Proposition \ref{prop4} for $N=2$ with $\lambda|v|\geqslant\eta$, we have
\begin{equation}
\int_{|t|<|v|^{-\sigma_1}}I_1dt=O(|v|^{-1}).\label{eq15_1}
\end{equation}
$I_2$ is estimated by \eqref{eq13_2} of Lemma \ref{lem13}, we therefore have
\begin{equation}
I_2\leqslant C\langle vt\rangle^{-2}(1+|t|^{1-2\hat{\gamma}_{\rm D}}+|t|^{2-4\hat{\gamma}_{\rm D}}).\label{eq15_2}
\end{equation}
Recall the condition $1/4<\hat{\gamma}_{\rm D}\leqslant1/2$, and we then compute
\begin{gather}
\int_0^{|v|^{-\sigma_1}}\langle vt\rangle^{-2}t^{1-2\hat{\gamma}_{\rm D}}dt=|v|^{2\hat{\gamma}_{\rm D}-2}\int_0^{|v|^{1-\sigma_1}}\langle\tau\rangle^{-2}\tau^{1-2\hat{\gamma}_{\rm D}}d\tau=O(|v|^{2\hat{\gamma}_{\rm D}-2}),\label{eq15_3}\\
\int_0^{|v|^{-\sigma_1}}\langle vt\rangle^{-2}t^{2-4\hat{\gamma}_{\rm D}}dt=|v|^{4\hat{\gamma}_{\rm D}-3}\int_0^{|v|^{1-\sigma_1}}\langle\tau\rangle^{-2}\tau^{2-4\hat{\gamma}_{\rm D}}d\tau=O(|v|^{4\hat{\gamma}_{\rm D}-3}).\label{eq15_4}
\end{gather}
we also have, by \eqref{eq1_11}, \eqref{eq15_3}, and \eqref{eq15_4},
\begin{equation}
\int_{|t|<|v|^{-\sigma_1}}I_2dt=O(|v|^{-1}).\label{eq15_5}
\end{equation}
By the same computation of \eqref{eq1_5}, $I_3$ is
\begin{equation}
I_3\leqslant\int_0^1\|(\nabla_x V^{\rm s})(\theta x+vt+e_1t^2/2)F(|x|<3\lambda|v||t|)\|d\theta\times\|xe^{-{\rm i}t|p|^2/2}M_{{\rm D}, v}(t)\Phi_0\|.\label{eq15_6}
\end{equation}
We here note that
\begin{equation}
\|xe^{-{\rm i}t|p|^2/2}M_{{\rm D}, v}(t)\Phi_0\|=\|(x+pt)M_{{\rm D}, v}(t)\Phi_0\|\leqslant C(1+|t|^{1-2\hat{\gamma}_{\rm D}}+|t|)\label{free_dynamics4}
\end{equation}
holds by \eqref{heisenberg_x} and \eqref{eq13_1}. From \eqref{eq1_8}, \eqref{eq15_6}, and \eqref{free_dynamics4}, $I_3$ is estimated such that
\begin{equation}
I_3\leqslant C\langle vt\rangle^{-1-\alpha}(1+|t|^{1-2\hat{\gamma}_{\rm D}}+|t|).
\end{equation}
To compute the following integral
\begin{equation}
\int_0^{|v|^{-\sigma_1}}\langle vt\rangle^{-1-\alpha}t^{1-2\hat{\gamma}_{\rm D}}dt\leqslant|v|^{-1-\alpha}\int_0^{|v|^{-\sigma_1}}t^{-\alpha-2\hat{\gamma}_{\rm D}}dt=O(|v|^{-1-\alpha}),\label{eq15_7}
\end{equation}
we can assume that $\alpha+2\hat{\gamma}_{\rm D}<1$ without loss of generality. Therefore, we have, by \eqref{eq1_10}, \eqref{eq1_11} and \eqref{eq15_7},
\begin{equation}
\int_{|t|<|v|^{-\sigma_1}}I_3dt=O(|v|^{-1}).\label{eq15_8}
\end{equation}
We next consider the integral on $|t|\geqslant|v|^{-\sigma_1}$. As in the proof of Proposition \ref{prop6}, take $0<\sigma_2<1$ which is also independent of $t$ and $v$. Its lower and upper bounds shall be determined at the end of this proof. We also put $I_4$, $I_5$ and $I_6$ such that
\begin{equation}
I\leqslant I_4+I_5+I_6,
\end{equation}
where
\begin{eqnarray}
I_4\Eqn{=}C_1\|F(|x|\geqslant3|v|^{\sigma_2}|t|)e^{-{\rm i}t|p|^2/2}f(p)F(|x|\leqslant|v|^{\sigma_2}|t|)\|,\\
I_5\Eqn{=}C_2\|F(|x|>|v|^{\sigma_2}|t|)\langle x\rangle^{-2}\|\|\langle x\rangle^2M_{{\rm D}, v}(t)\Phi_0\|,\\
I_6\Eqn{=}\|\{V^{\rm s}(x+vt+e_1t^2/2)-V^{\rm s}(vt+e_1t^2/2)\}\nonumber\\
\Eqn{}\hspace{20mm}\times F(|x|<3|v|^{\sigma_2}|t|)e^{-{\rm i}t|p|^2/2}M_{{\rm D}, v}(t)\Phi_0\|.
\end{eqnarray}
$I_4$ is the same as in the proof of Proposition \ref{prop6}, and we thus have
\begin{equation}
\int_{|t|\geqslant|v|^{-\sigma_1}}I_4dt=O(|v|^{\sigma_1-2\sigma_2})\label{eq15_9}
\end{equation}
by using Proposition \ref{prop4} for $N=2$ with $|v|^{\sigma_2}\geqslant\eta$. Similar to \eqref{eq15_2}, with \eqref{eq13_2} of Lemma \ref{lem13}, $I_5$ is
\begin{equation}
I_5\leqslant C\langle|v|^{\sigma_2}t\rangle^{-2}(1+|t|^{1-2\hat{\gamma}_{\rm D}}+|t|^{2-4\hat{\gamma}_{\rm D}}).
\end{equation}
We compute the following integrals
\begin{gather}
\int_{|v|^{-\sigma_1}}^\infty\langle|v|^{\sigma_2}t\rangle^{-2}t^{1-2\hat{\gamma}_{\rm D}}dt\leqslant|v|^{-2\sigma_2}\int_{|v|^{-\sigma_1}}^\infty t^{-1-2\hat{\gamma}_{\rm D}}dt=O(|v|^{2\sigma_1\hat{\gamma}_{\rm D}-2\sigma_2})\label{eq15_10},\\
\int_{|v|^{-\sigma_1}}^\infty\langle|v|^{\sigma_2}t\rangle^{-2}t^{2-4\hat{\gamma}_{\rm D}}dt\leqslant|v|^{-2\sigma_2}\int_{|v|^{-\sigma_1}}^\infty t^{-4\hat{\gamma}_{\rm D}}dt=O(|v|^{\sigma_1(4\hat{\gamma}_{\rm D}-1)-2\sigma_2})\label{eq15_11}.
\end{gather}
Here $\hat{\gamma}_{\rm D}>1/4$ was used in \eqref{eq15_11}. Note that $4\hat{\gamma}_{\rm D}-1<2\hat{\gamma}_{\rm D}<1$ holds for $1/4<\hat{\gamma}_{\rm D}<1/2$ and, with the computations of \eqref{eq1_14}, \eqref{eq15_10}, and \eqref{eq15_11}, we then have
\begin{equation}
\int_{|t|\geqslant|v|^{-\sigma_1}}I_5dt=O(|v|^{\sigma_1-2\sigma_2}).\label{eq15_12}
\end{equation}
As in \eqref{eq15_6}, $I_6$ is
\begin{equation}
I_6\leqslant\int_0^1\|(\nabla_x V^{\rm s})(\theta x+vt+e_1t^2/2)F(|x|<3|v|^{\sigma_2}|t|)\|d\theta\times\|xe^{-{\rm i}t|p|^2/2}M_{{\rm D}, v}(t)\Phi_0\|.\label{eq15_13}
\end{equation}
Put $I_{6,1}$ and $I_{6.2}$ such that, with the estimate \eqref{free_dynamics4},
\begin{equation}
I_6\leqslant I_{6,1}+I_{6,2},
\end{equation}
where
\begin{eqnarray}
I_{6,1}\Eqn{=}C\int_0^1\|(\nabla_x V^{\rm s})(\theta x+vt+e_1t^2/2)F(|x|<3|v|^{\sigma_2}|t|)\|d\theta\times(1+|t|),\\
I_{6,2}\Eqn{=}C\int_0^1\|(\nabla_x V^{\rm s})(\theta x+vt+e_1t^2/2)F(|x|<3|v|^{\sigma_2}|t|)\|d\theta\times|t|^{1-2\hat{\gamma}_{\rm D}}.\qquad
\end{eqnarray}
The estimate of $I_{6,1}$ was already obtained in \eqref{eq1_26} in the proof of Proposition \ref{prop6}, we thus have
\begin{equation}
\int_{|t|\geqslant|v|^{-\sigma_1}}I_{6,1}dt=O(|v|^{-1-\alpha+\sigma_1(2+3\alpha)}).\label{eq15_14}
\end{equation}
As for $I_{6,2}$, by using the H\"older inequality again, we estimate the integral of $I_{6,2}$ such that
\begin{equation}
\int_{|v|^{-\sigma_1}}^\infty I_{6,2}dt\leqslant\left(\int_{|v|^{-\sigma_1}}^\infty I_{6,2}^{q_1}dt\right)^{1/q_1}\left(\int_{|v|^{-\sigma_1}}^\infty I_{6,2}^{q_2}dt\right)^{1/q_2},\label{eq15_15}
\end{equation}
where $q_1$ and $q_2$ are the H\"older conjugates of each other. From the estimate
\begin{equation}
|x+vt+e_1t^2/2|\geqslant c_1|v||t|
\end{equation}
of \eqref{eq1_18} and assumption \eqref{short_range}, we compute
\begin{eqnarray}
\lefteqn{\int_{|v|^{-\sigma_1}}^\infty I_{6,2}^{q_1}dt\leqslant C\int_{|v|^{-\sigma_1}}^\infty\langle vt\rangle^{-q_1(1+\alpha)}t^{q_1(1-2\hat{\gamma}_{\rm D})}dt}\nonumber\\
\Eqn{}\leqslant C|v|^{-q_1(1+\alpha)}\int_{|v|^{-\sigma_1}}^\infty t^{-q_1(\alpha+2\hat{\gamma}_{\rm D})}dt=O(|v|^{-q_1(1+\alpha)-\sigma_1\{1-q_1(\alpha+2\hat{\gamma}_{\rm D})\}}).\quad\label{eq15_16}
\end{eqnarray}
Although we assumed that $\alpha+2\hat{\gamma}_{\rm D}<1$ in the estimate \eqref{eq15_7}, we choose $q_1>1$ which satisfies $q_1(\alpha+2\hat{\gamma}_{\rm D})>1$ in \eqref{eq15_16}. Again, from
\begin{equation}
|x+vt+e_1t^2/2|\geqslant c_2t^2
\end{equation}
of \eqref{eq1_18} and \eqref{short_range}, we compute
\begin{equation}
\int_{|v|^{-\sigma_1}}^\infty I_{6,2}^{q_2}dt\leqslant C\int_{|v|^{-\sigma_1}}^\infty\langle t^2\rangle^{-q_2(1+\alpha)}t^{q_2(1-2\hat{\gamma}_{\rm D})}dt=O(|v|^{-\sigma_1\{1-q_2(1+2\alpha+2\hat{\gamma}_{\rm D})\}}).\label{eq15_17}
\end{equation}
We here used $-2q_2(1+\alpha)+q_2(1-2\hat{\gamma}_{\rm D})=-q_2(1+2\alpha+2\hat{\gamma}_{\rm D})<-1$. By \eqref{eq15_15}, \eqref{eq15_16}, \eqref{eq15_17}, and
\begin{gather}
-1-\alpha-\sigma_1\{1-q_1(\alpha+2\hat{\gamma}_{\rm D})\}/q_1-\sigma_1\{1-q_2(1+2\alpha+2\hat{\gamma}_{\rm D})\}/q_2\nonumber\\
=-1-\alpha+\sigma_1(3\alpha+4\hat{\gamma}_{\rm D}),
\end{gather}
we have
\begin{equation}
\int_{|t|\geqslant|v|^{-\sigma_1}}I_{6,2}dt=O(|v|^{-1-\alpha+\sigma_1(3\alpha+4\hat{\gamma}_{\rm D})}).\label{eq15_18}
\end{equation}
Because $\hat{\gamma}_{\rm D}<1/2$, \eqref{eq15_14} and \eqref{eq15_18} state that
\begin{equation}
\int_{|t|\geqslant|v|^{-\sigma_1}}I_6dt=O(|v|^{-1-\alpha+\sigma_1(2+3\alpha)}).\label{eq15_19}
\end{equation}
By combining \eqref{eq15_1}, \eqref{eq15_5}, \eqref{eq15_8}, \eqref{eq15_9}, \eqref{eq15_12}, and \eqref{eq15_19}, we finally obtain
\begin{equation}
\int_{-\infty}^\infty Idt=O(|v|^{-1})+O(|v|^{\sigma_1-2\sigma_2})+O(|v|^{-1-\alpha+\sigma_1(2+3\alpha)}).\label{eq15_20}
\end{equation}
As in the proof of Proposition \ref{prop6}, by determining the size of $\sigma_1$ and $\sigma_2$ such that $0<\sigma_1<\alpha/(2+3\alpha)$ and $(1+\sigma_1)/2<\sigma_2<1$, \eqref{eq15_20} completes our proof.
\end{proof}

The following propagation estimate for the long-range part $V^{\rm l}\in\hat{\mathscr{V}}_{\rm D}^{\rm l}$ was given in \cite[Lemma 4.5]{AdFuIs} for $\mu=0$.
\begin{Prop}\label{prop16}
Let $v$ and $\Phi_v$ be as in Theorem \ref{the8}. Then
\begin{equation}
\int_{-\infty}^\infty\|\{V^{\rm l}(x)-V^{\rm l}(pt-e_1t^2/2)\}e^{-{\rm i}tH_0^{\rm S}}M_{\rm D}(t)\Phi_v\|dt=O(|v|^{1-4\hat{\gamma}_{\rm D}+\epsilon_2})\label{estimate_V^l_long_range_dollard}
\end{equation}
holds with any small $\epsilon_2>0$ as $|v|\rightarrow\infty$ for $V^{\rm l}\in\hat{\mathscr{V}}_{\rm D}^{\rm l}$.
\end{Prop}

For $V^{\rm l}\in\mathscr{V}_{\rm G}^{\rm l}\cup\hat{\mathscr{V}}_{\rm D}^{\rm l}$, we introduce the auxiliary 
Dollard-Graf-type modified wave operators
\begin{equation}
\Omega_{{\rm D},{\rm G},v}^\pm=\slim_{t\rightarrow\pm\infty}e^{{\rm i}tH^{\rm S}}U_{{\rm D},{\rm G},v}(t),\quad U_{{\rm D},{\rm G},v}(t)=e^{-{\rm i}tH_0^{\rm S}}M_{\rm D}(t)M_{{\rm G},v}(t)\label{dollard_garf_modifier}
\end{equation}
according to \cite{AdKaKaTo} and \cite{AdFuIs}. Because the Dollard-type modified wave operators \eqref{wave_operators_dollard} and limits \eqref{I_G_v} exist, we see that $\Omega_{{\rm D},{\rm G},v}^\pm=W_{\rm D}^\pm I_{{\rm G},v}^\pm(t)$.\par

By virtue of Propositions \ref{prop10}, \ref{prop11}, \ref{prop12}, \ref{prop14}, \ref{prop15}, and \ref{prop16}, the following corollary is proved as in Corollary \ref{cor7}. We therefore omit its proof.
\begin{Cor}\label{cor17}
Let $v$ and $\Phi_v$ be as in Theorem \ref{the8}. Then, for $V^{\rm vs}\in\mathscr{V}^{\rm vs}$, $V^{\rm s}\in\mathscr{V}^{\rm s}$, and $V^{\rm l}\in\mathscr{V}_{\rm G}^{\rm l}\cup\hat{\mathscr{V}}_{\rm D}^{\rm l}$,
\begin{equation}
\|\{e^{-{\rm i}tH^{\rm S}}\Omega_{{\rm D},{\rm G},v}^\pm-U_{{\rm D},{\rm G},v}(t)\}\Phi_v\|=
\begin{cases}
\ O(|v|^{-1/2-\epsilon_1}) & \mbox{\rm if}\ V^{\rm l}\in\mathscr{V}_{\rm G}^{\rm l},\\
\ O(|v|^{1-4\hat{\gamma}_{\rm D}+\epsilon_2}) & \mbox{\rm if}\ V^{\rm l}\in\hat{\mathscr{V}}_{\rm D}^{\rm l}\\
\end{cases}
\end{equation}
holds as $|v|\rightarrow\infty$ uniformly in $t\in\mathbb{R}$, where $\epsilon_1$ and $\epsilon_2$ are taken in Propositions \ref{prop11} and \ref{prop16} respectively.
\end{Cor}

We now prove Theorem \ref{the8}.
\begin{proof}[Proof of Theorem \ref{the8}]
This proof is almost similar to the proof of Theorem \ref{the3} (see also \cite[Theorem 3.5]{We} and \cite[Theorem 4.1]{AdFuIs}). We give here a sketch. Note that the Dollard-type modified scattering operator $S_{\rm D}$ is represented such that
\begin{equation}
S_{\rm D}=I_{{\rm G},v}(\Omega_{{\rm D},{\rm G},v}^+)^*\Omega_{{\rm D},{\rm G},v}^-
\end{equation}
and put $V_{v,t}^{\rm l}$ by
\begin{equation}
V_{v,t}^{\rm l}=V^{\rm vs}(x)+V^{\rm s}(x)-V^{\rm s}(vt+e_1t^2/2)+V^{\rm l}(x)-V^{\rm l}(pt-e_1t^2/2),
\end{equation}
then we have
\begin{gather}
{\rm i}(S_{\rm D}-I_{{\rm G},v})\Phi_v={\rm i}I_{{\rm G},v}(\Omega_{{\rm D},{\rm G},v}^+-\Omega_{{\rm D},{\rm G},v}^-)^*\Omega_{{\rm D},{\rm G},v}^-\Phi_v\nonumber\\
=I_{{\rm G},v}\int_{-\infty}^\infty U_{{\rm D},{\rm G},v}(t)^*V_{v,t}^{\rm l}e^{-{\rm i}tH^{\rm S}}\Omega_{{\rm D},{\rm G},v}^-\Phi_vdt.
\end{gather}
We therefore obtain, by the relations $[S_{\rm D}, p_j]=[S_{\rm D}-I_{{\rm G},v},p_j-v_j]$ and $(p_j-v_j)\Phi_v=(p_j\Phi_0)_v$,
\begin{equation}
|v|({\rm i}[S_{\rm D},p_j]\Phi_v,\Psi_v)=I_{{\rm G},v}\{I_{\rm D}(v)+R_{\rm D}(v)\}
\end{equation}
with
\begin{eqnarray}
I_{\rm D}(v)\Eqn{=}|v|\int_{-\infty}^\infty\Big\{(V_{v,t}^{\rm l}U_{{\rm D},{\rm G},v}(t)(p_j\Phi_0)_v,U_{{\rm D},{\rm G},v}(t)\Psi_v)\nonumber\\
\Eqn{}\hspace{20mm}-(V_{v,t}^{\rm l}U_{{\rm D},{\rm G},v}(t)\Phi_v,U_{{\rm D},{\rm G},v}(t)(p_j\Psi_0)_v)\Big\}dt,\\
R_{\rm D}(v)\Eqn{=}|v|\int_{-\infty}^\infty(\{e^{-{\rm i}tH^{\rm S}}\Omega_{{\rm D},{\rm G},v}^--U_{{\rm D},{\rm G},v}(t)\}(p_j\Phi_0)_v,V_{v,t}^{\rm l}U_{{\rm D},{\rm G},v}(t)\Psi_v)dt\nonumber\\
\Eqn{}-|v|\int_{-\infty}^\infty(\{e^{-{\rm i}tH^{\rm S}}\Omega_{{\rm D},{\rm G},v}^--U_{{\rm D},{\rm G},v}(t)\}\Phi_v,V_{v,t}^{\rm l}U_{{\rm D},{\rm G},v}(t)(p_j\Psi_0)_v)dt.\qquad
\end{eqnarray}
If $V^{\rm l}\in\mathscr{V}_{\rm G}^{\rm l}$, Propositions \ref{prop10}, \ref{prop11}, \ref{prop12}, and Corollary \ref{cor17} imply that
\begin{equation}
R_{\rm D}(v)=O(|v|^{-2\epsilon_1});\label{long_range_remainder_graf}
\end{equation}
otherwise, if $V^{\rm l}\in\mathscr{V}_{\rm D}^{\rm l}$, Propositions \ref{prop14}, \ref{prop15}, \ref{prop16}, and Corollary \ref{cor17} imply that
\begin{equation}
R_{\rm D}(v)=O(|v|^{3-8\gamma_{\rm D}+2\epsilon_2}).\label{long_range_remainder_dollard}
\end{equation}
We can take arbitrarily small $\epsilon_2>0$, and the condition $\gamma_{\rm D}>3/8$ is required to guarantee the convergence of \eqref{long_range_remainder_dollard}. We therefore conclude that
\begin{equation}
\lim_{|v|\rightarrow\infty}R_{\rm D}(v)=0.
\end{equation}
\par
Because the rest of this proof is the same as in \cite{We} and \cite{AdFuIs}, we omit it here.
\end{proof}

Theorem \ref{the3} is proved similarly to \cite[Theorem 1.2]{We} (see also \cite{EnWe}) and the proof is also omitted.

\bigskip
\noindent\textbf{Acknowledgments.} This study was partially supported by the Grant-in-Aid for Young Scientists (B) \#16K17633 from JSPS.


\end{document}